\documentclass[sub]{jalc}

\usepackage[utf8]{inputenc}
\usepackage[english]{babel}
\usepackage{amsmath}
\usepackage{amsfonts}
\usepackage{amssymb}
\usepackage{graphicx}
\usepackage{tikz}
\usetikzlibrary{automata, positioning, arrows, calc,shapes}

\begin{document}

\title{A Characterization of Totally Compatible Automata}

\runningauthors{\textsc{D.~Casas}}
\author[EKB]{David Fernando Casas Torres}
\address[EKB]{Institute of Natural Sciences and Mathematics,
  Ural Federal University,\\
  Pr. Lenina 51,
  620000 Ekaterimburg,
  Russia\\
  \email{dkasastorres@urfu.ru}
 }

\date{}

\maketitle
\begin{abstract}
Every function on a finite set defines an equivalence relation and, therefore, a partition called the kernel of the function. Automata such that every possible partition is the kernel of a word are called totally compatible.  A characterization of such automata is given together with an algorithm to recognize them in polynomial running time with respect to the number of states. 
\end{abstract}
\section{Background and Motivation}

A \emph{Deterministic Finite Automaton} (DFA), here simply called an automaton, is a triple \(\mathcal{A} = (Q,\Sigma, \delta)\). Both \(Q\) and \(\Sigma\) are finite sets and \(\delta \colon Q \times \Sigma \rightarrow Q\) is a function that describes the action of the \emph{alphabet} \(\Sigma\) on the \emph{state} set \(Q\). Fixing a letter \(a \in \Sigma\) in the second argument of \(\delta\) produces a transformation \(\delta_a \colon Q \rightarrow Q\) that is defined as \(\delta_a(q) : = \delta(q,a) \); here it will be denoted with right notation, i.e., \(\delta_a(q) \: = q \cdot a\). Let \(\Sigma^\ast\) be the set of all words of the alphabet \(\Sigma\) (including the empty word, denoted by \(\varepsilon\)). The action of a word \(w \in \Sigma^\ast\) over \(Q\) can be defined recursively: \(q\cdot \varepsilon := q\) for every state \(q\); if \(w = va\) with \(v \in \Sigma^\ast\) and \(a \in \Sigma\) then for every state \(q \in Q\), \(q \cdot w := (q \cdot v) \cdot a\).  That is why, when discussing automata, we consider just their states and alphabet. Note how any word \(w \in \Sigma^\ast\) produces a transformation over the set of states, this is the composition of the transformations defined by each letter.

If \(P \subseteq Q\) is a  nonempty set of states, its image under \(w\) is the set 
\[P \cdot w := \{p \cdot w \mid p \in P\}.\]
 The word \(w \in \Sigma^\ast\) \emph{synchronizes} a set of states \(P \subseteq Q\) if \(\vert P \cdot w\vert = 1\), i.e., the transformation sends every state of \(P\) to a single state. An automaton for which there is a word that synchronizes the set of all states, i.e., \(\vert Q \cdot w \vert  = 1\),  is called synchronizing; such word is called a \emph{reset} word.\\
 
  Given a transformation or, what is the same, a word, \(w\) over  \(Q\), consider the following relation:
\[\ker(w) := \{(p,q) \in Q \times Q \mid p \cdot w = q \cdot w\}.\]
It is easy to see that this is an equivalence relation, which is called the \emph{kernel} of the transformation \(w\). The equivalence relation that produces an unique class (all the states are related) and the one that produces one class for each state (each element is related just with itself) are called the trivial equivalence relations.\\  

In \cite{bondar2016completely,bondar2018characterization} it was considered the reachability of non-empty subsets of states. Let \(\mathcal{A} = (Q,\Sigma)\) be an automaton, the subset \(P \subseteq Q\) is \emph{reachable} if there is a word \(w \in \Sigma^\ast\) such that \(Q \cdot w = P\).  The automaton \(\mathcal{A}\) is \emph{completely reachable} if every non-empty subset \(P \subseteq Q\)  is reachable. In \cite{bondar2018characterization} a condition to characterize these automata was given. This characterization depended on the strong connectivity of a certain  graph constructed recursively. It is worth noting that this characterization does not lead, however, to a polynomial time algorithm for recognizing complete reachabilty and no such algorithm is known so far. A completely reachable automaton has the capacity of ``realizing'' every non-empty subset with a word. Inspired by this it is natural to consider the dual concept: partitions. While the images of words define subsets, their kernels define partitions and it is  possible to look for automata that realize every partition of the state set. This means that we focus not on the image set of the transformation( the \textit{``right side"}), but on the partition  defined by the transformation over the domain (the \textit{``left side''}).\\

In this paper we will consider the class of automata such that every possible equivalence relation in the set of states can be obtained from a word. First we properly define them and give some examples. In Section 3 a characterization of these automata is given. Following this characterization, we will show an algorithm to identify these automata; then proving that this property is decidable in polynomial time on the number of states.

\section{Definition and Examples}

Let \(\mathcal{A} = (Q, \Sigma)\) be an automaton and \(\rho\) an arbitrary equivalence relation over the set \(Q\). A word \(w \in \Sigma^\ast \) is \emph{compatible} with \(\rho\) if and only if \(\ker(w) = \rho\). From this definition the following remark can be easily seen:

\begin{remark}
Let \(\mathcal{A} = (Q,\Sigma)\) be an automaton and \(\rho\) an equivalence relation over \(Q\) with \(R_1, \dots, R_k\) (\(k \geq 1\)), its  correspondent equivalence classes. A word \(w \in \Sigma^\ast\) is compatible with \(\rho\) if and only if:
\begin{enumerate}
\item The word synchronizes each class, i.e., \(\vert R_i \cdot w \vert = 1\) for each \(1 \leq i \leq k\), and
\item the images by \(w\) of different classes are pairwise different, i.e., \(R_i \cdot w \neq R_j \cdot w\) for all \(1 \leq i \neq j \leq k\).
\end{enumerate}
\label{CharacterizingCompatible} 
\end{remark}

Note that any permutation (or the empty word) and a reset word are compatible with the trivial partitions, those made from just singletons and the whole set respectively.\\

\begin{definition}
Let \(\mathcal{A} = (Q, \Sigma)\) be an automaton. \(\mathcal{A}\) is \textbf{totally compatible} if for every equivalence relation \(\rho\) over the set \(Q\) there is a word \(w_\rho \in \Sigma^\ast\) compatible with this relation.
\end{definition}

A trivial example of a totally compatible automaton is one that its transformations generate the full transformation semigroup. Note that any totally compatible automaton is synchronizing.\\

The automaton \(\mathcal{T} = (\{1,2,3\}, \{c, t\})\), where \(c = (1,2,3)\), is the cyclic permutation, and \(1 \cdot t = 2 \cdot t = 1\) and \(3 \cdot t = 3\) (see figure \ref{AutomatonT}), is an example of a totally compatible automaton. In  the table \ref{Table3Partitions} we can see all the partitions of three elements and their compatible transformations in \(\mathcal{T}\).\\

\begin{table}
\centering
\begin{tabular}{cc}
\hline
Partition & Compatible Word\\
\hline
\(1\vert 2 \vert 3\) & \(c\)\\
\hline
\(12 \vert 3\) & \(t\)\\
\hline
 \(13 \vert 2\) & \(ct\)\\
\hline 
\(1 \vert 23\) & \(cct\) \\
\hline
\(123\) & \(tct\) \\
\hline 
\end{tabular}
\caption{Partitions of three elements and their compatible words.}
\label{Table3Partitions}
\end{table}

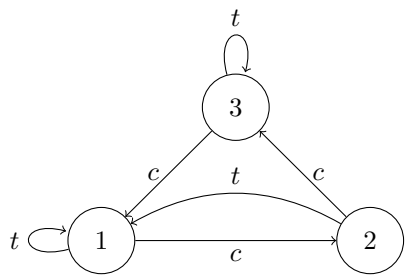
\begin{figure}
\centering
\begin{tikzpicture}[node distance =2.5cm]
	\node[state] (3) {\(3\)};
	\node[state, below left of= 3 ] (1) {\(1\)};
	\node[state, below right of= 3] (2) {\(2\)};
	
	\draw
		(1) edge[->,  below] node{\(c\)}(2)
		(2) edge[->, right] node{\(c\)}(3)
		(3) edge[->, left] node{\(c\)}(1); 
	\draw
		(1) edge[->, loop left, left] node{\(t\)}(1)
		(2) edge[->, bend right, above] node{\(t\)}(1)
		(3) edge[->, loop above, above] node{\(t\)}(3);
\end{tikzpicture}
\caption{The automaton \(\mathcal{T}\).}
\label{AutomatonT}
\end{figure}

Recall the sequence of \v{C}ern\'y automata \(\mathcal{C}_n\), with \(n \geq 2\). Each \(\mathcal{C}_n\) has \(n\) states, \(\{1, \dots, n\}\) and two letters \(\{a,b\}\); where \(a\) acts as the complete cycle in the usual order and \(b\) fixes every state except \(1\) that is sent to 2.
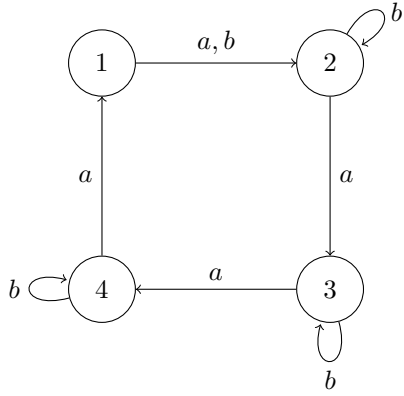
\begin{figure}[ht]
	\centering
	\begin{tikzpicture}
	\node[state] (1) {\(1\)};
	\node[state, right of=1, xshift = 2cm] (2) {\(2\)};
	\node[state, below of=2, yshift = -2cm] (3) {\(3\)};
	\node[state, below of=1, yshift = -2cm] (4) {\(4\)};
	
	\draw 
		(1) edge[->, above] node{\(a, b\)} (2)
		(2) edge[->, right] node{\(a\)} (3)
		(3) edge[->, above] node{\(a\)} (4)
		(4) edge[->, left] node{\(a\)} (1);
	\draw
		(2) edge[->, in=30, out =60, loop, right] node{\(b\)} (2)
		(3) edge[->, loop below] node{\(b\)} (3)
		(4) edge[->, loop left]  node{\(b\)} (4);
	\end{tikzpicture}
\caption{The automaton \(\mathcal{C}_4\).}
\end{figure}
These are examples of completely reachable automata (\cite{bondar2016completely}) but not totally compatible. Since there is not a permutation (a multiple of the cycle) that sends the set \(\{1,3\}\) to \(\{1,2\}\), then there is no word that synchronizes \(\{1,3\}\) and acts as a bijection on the rest of the states, thus there is no  compatible word for the partition \(13 \vert 2 \dots n\).

\section{A Characterization of Totally Compatible Automata}

Our characterization of totally compatible automata depends on the existence of a special kind of transformations. A word \(w\) is said to have \emph{deficiency 1} (or to be \emph{1-defect}) if \(\vert Q \cdot w \vert = \vert Q \vert - 1\). This means that there is exactly one state without pre-image by \(w\). Note that in order to \(w\) be 1-defect it must synchronize two different states, say \(p,q \in Q\), and it musts act as an injection on the other states, hence generating a kernel of size \(\vert Q \vert - 1\) where one equivalence class is the set \(\{p,q\}\) and the other classes are singletons.\\ 

\begin{proposition}
Let \(\mathcal{A} = (Q,\Sigma)\) be an automaton. \(\mathcal{A}\) is totally compatible if and only if for every pair of states \(p,q \in Q\) there is a 1-defect word \(w \in \Sigma^\ast\) such that \(p \cdot w = q \cdot p\).
\end{proposition}

\begin{proof}
Let \(\mathcal{A} = (Q,\Sigma)\) be an automaton and \(p,q \in Q\) two arbitrary different states. If \(\mathcal{A}\) is totally compatible, then there is a word in \(\Sigma^\ast\) compatible with the equivalence relation generated by adding the pairs \((p,q)\) and \((q,p)\) to the identity relation. It is trivial that this is a 1-defect word synchronizing these states.\\

If \(\mathcal{A}\) meets the condition, let \(\vert Q \vert = n\). It will be shown that, for any  \(1 \leq k < n-1\), if any equivalence relation of size (the number of different equivalence classes) bigger than \(k\) has a compatible word, then a word compatible with any equivalence relation of size \(k\) can be constructed. This will prove the result since the empty word (or identity transformation) is compatible with the trivial equivalence of size \(n\), and for hypothesis all the equivalence relations of size  \(n-1\) have compatible words. \\

Suppose that any equivalence relation of size  strictly bigger than \(k\) has a compatible word in \(\mathcal{A}\) and let \(\rho\) be an arbitrary equivalence relation over \(Q\) of size \(k\). Let \(R_1, \dots, R_{k}\) be the equivalent classes of \(\rho\). Without loss of generality suppose that \(\vert R_{k} \vert > 1\). Now, let \(\pi\) be an equivalence relation over \(Q\) such that  all its equivalence classes coincide with those of \(\rho\) except in the case of \(R_k\), which is divided in two different classes, i.e. \(R_{k} = P_{k} \sqcup P_{k +1}\) with \(P_k \cap P_{k+1} = \emptyset\) and neither of both sets are empty. This is possible as \(R_k\) has more than one element.\\

The size of \(\pi\) is \(k +1 \) and from the assumption there is a word \(w_\pi \in \Sigma^\ast\) compatible with \(\pi\). Note that both conditions of the remark \ref{CharacterizingCompatible} are met by \(w_\pi\) with almost all equivalence classes of \(\rho\) except for \(R_k\), where \(R_{k} \cdot w_\pi = (P_{k} \sqcup P_{k +1}) \cdot w_\pi = \{p,q\}\) and \(p \neq q\). For hypothesis, there is a 1-defect word \(v \in \Sigma^\ast\) such that \(p \cdot v = q \cdot v\), and that acts as an injection in the rest of states.  Observe that \(w_\pi v\) synchronizes \(R_k\) and hence meets both conditions of remark \ref{CharacterizingCompatible}, making it compatible with \(\rho\).
\end{proof}

This characterization suggests that in order to decide if an arbitrary automaton is totally compatible it is needed to concentrate the attention on the letters that act as permutations and as 1-defect transformations. Since these transformations have the biggest ranks (cardinal of the image set), \(n\) and \(n-1\) respectively, and the composition of transformations does not increment the rank, it would not be possible to obtain  1-defects using transformations with lower ranks.\\

Let \(\mathcal{A} = (Q,\Sigma)\) be an automaton and \(w \in \Sigma^\ast\) a 1-defect word; there is an state which does not have a pre-image by the action of \(w\), and one state that has two, call them the \emph{excluded} and \emph{duplicated} states of \(w\) respectively and denote them \(\text{exc}(w)\) and \(\text{dup} (w)\) or, if there is not place to confusions, \(e_w\) and \(d_w\). Call the set of pre-images of the duplicate state the \emph{root} of \(w\) and denote it \(\text{root}(w)\) or \(R_w\). In these terms \(\mathcal{A}\) is totally compatible if and only if every 2-subset (subset of cardinality 2) is the root of a 1-defect word.\\

Given two 1-defect words of \(\mathcal{A}\), \(w\) and \(v\), it is easy to see that the word \(wv\) keeps being 1-defect if and only if \(e_w \in R_v\), and in this case the root of \(wv\) is the same as the one of \(w\); therefore the concatenation of two 1-defect words that produce a 1-defect will not create new roots.\\

Let \(\mathcal{A} =(Q,\Sigma)\) be an automaton, \(\Sigma_0 \) and \(\Sigma_1\) be the set of letters that act as permutations and 1-defects over \(Q\) respectively. Consider the directed and labeled graph \(\mathcal{A}^{[2]}_0\) with vertex set \(Q^{[2]}\), the set of 2-subsets of \(Q\), and with \((P,T)\) as an edge with label \(\sigma \in \Sigma_0\), if and only if \(P\cdot \sigma = T\). Denote by \(\mathfrak{R} \subseteq Q^{[2]}\), the subset that contains all the roots of every 1-defect letter in \(\Sigma_1\). If, from an arbitrary \(P \in Q^{[2]}\) not in \(\mathfrak{R}\), there is a path to a root set \(R_a \in \mathfrak{R}\), and this path is labeled by the sequence \(\sigma_1, \sigma_2, \dots, \sigma_k\) (\(k \geq 1\)), it can be  concluded that: the word \(w = (\sigma_1 \sigma_2 \dots \sigma_k) \, a\) is 1-defect. It is the concatenation of a permutation and a 1-defect; and that \({\rm{root}}(w) = P\), because \(P \cdot (\sigma_1 \dots \sigma_k)= R_a \) and \(P \cdot w = \{d_a\}\).\\
From the characterization of totally compatible automata and the previous discussion the procedure to decide if \(\mathcal{A}\) is totally compatible goes as follows:

\begin{enumerate}
\item Construct the graph \(\mathcal{A}^{[2]}_0\),
\item for every 2-subset, \(P\), check if there is a path that connects \(P\) to the set \(\mathfrak{R}\). If this is the case conclude that \(\mathcal{A}\) is totally compatible.
\end{enumerate}

If the automaton \(\mathcal{A}\) has \(n\) states. The first step can be made in \(O(n^2\vert\Sigma_0 \vert)\) time, since there are \({n \choose 2}\) possible 2-subsets and for each one it is needed to determine its image by each permutation in \(\Sigma_0\). If the set \(\mathfrak{R}\) is considered as one single vertex in \(\mathcal{A}^{[2]}_0\), a Breadth-First-Search algorithm can solve the second step in time \(O((n^2 + \vert \Sigma_0 \vert)n)\). Thus, the decision problem of  totally compatibility can be solved in polynomial time.

\begin{proposition}
Totally compatible automata are \textbf{P}-decidable.
\end{proposition}

The previous result poses a contrast with  complete reachability, since, until the moment, it is not clear the complexity of deciding this property. This situation is illustrated in Figure \ref{Compleixties}. There it is considered some relevant kind of automata and their decision complexities, the edges from  bottom to top represent containment relation. \\

\begin{figure}
\begin{tikzpicture}
	\node[] (1) {\small{Synchronizing (\textbf{P})}};
	\node[below right =of 1] (2) {\small{Complete  Reachable (?)}};
	\node[below left =of 1] (3) {\small{Totally  Compatible (\textbf{P})}};
	\node[below right =of 3] (4) {\small{C.R \& T.C (?)}};
	\node[below =of 4] (5) {\small{With full transformation monoid (\textbf{P})}};
	
	\draw
		(2) edge (1)
		(3) edge (1)
		(4) edge (2)
		(4) edge (3)
		(5) edge (4);
\end{tikzpicture}
\caption{Inclusions between and decision complexities of classes of automata. }
\label{Compleixties}
\end{figure}
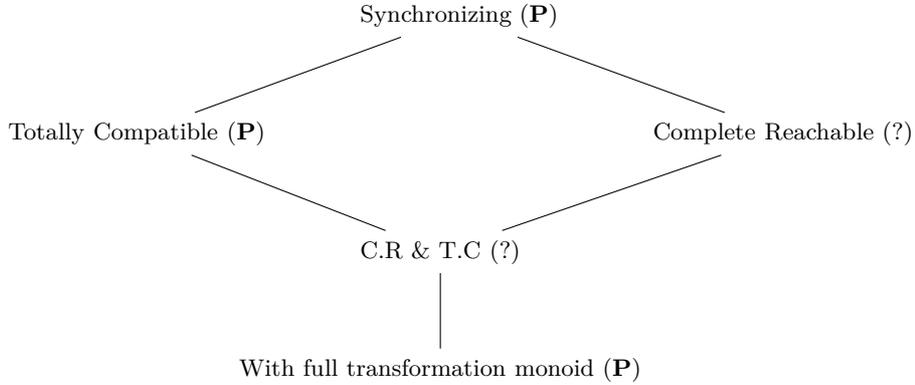

A consequence of the characterization and the fact that permutation groups which connect all possible 2-subsets of states (2-homogeneous) require at least two generators \cite{araujo2019orbits} is that  there can not be totally compatible automata with more than three states and less than three letters.

\begin{corollary}
If \(\mathcal{A} = (Q, \Sigma)\) is an totally compatible automaton and \(\vert Q \vert > 3\) then \(\vert \Sigma \vert > 2\).
\end{corollary}

Let \(\mathcal{A} = (Q, \Sigma)\) be an automaton. It is said that \(\mathcal{A}\) is a minimal completely reachable automaton if the transformation semigroup generated by \(\Sigma\) has exactly \(2^n - 1\) different transformations, thus one unique transformation for every non-empty subset. In the same sense, define \(\mathcal{A}\) as a \emph{minimal totally compatible automaton} if for every possible equivalence relation over \(Q\), there is only one compatible transformation in the generated semigroup; therefore the semigroup generated by the letters has exactly \(\mathcal{B}_{\vert Q \vert}\) transformations, where \(\mathcal{B}_n\) stands for the Bell's number of a set with \(n\) elements. In \cite[Section 12.4]{Ganyushkin2008ClassicalFT} it is presented the construction of a minimal totally compatible automaton with \(n \in \mathbb{N}\) states. In the terminology used there this automata is an \(\mathcal{L}\)-cross section in the full transformation monoid, moreover  it is proven that, up to isomorphism, this automaton is unique. From the construction it can be derived that this automaton requires  at least \({n \choose 2}\) letters, one 1-defect for each possible 2-subset of a set with \(n\) states. This automaton is an example of a totally compatible automaton which is not completely reachable since the images of all 1-defect generators are the same and it has no permutations letter.\\

\section{Conclusion}

Inspired by the notion of automata that can obtain every possible non-empty subset of states, in this note it was investigated the dual notion; characterizing automata that can obtain every possible partition of the states set. It was, also, proposed an algorithm that runs in polynomial time on the number of states to identify this kind of automata.\\

The research of completely reachable automata was, partly, conducted to further study the synchronization problem and the \v{C}ern\'y conjecture. It is trivial to see that if an automaton is totally compatible is synchronizable since there must be a transformation for the complete equivalence relation. Thus, it is natural to ask for a bound of the shortest synchronization word of a totally compatible automaton. If a totally compatible automaton has just one letter of defect 1, then the group generated by its permutation letter must be 2-homogeneous and therefore primitive \cite{beaumont1955set}. Although not explicitly stated in \cite[Remark 8]{gonze2019interplay} it is proven that automata where its permutation letters generate a 2-homogeneous group have a synchronization word of size at most quadratic on the size of states. Nevertheless the previously mentioned, it is yet open to give an upper bound to the length of the shortest word that synchronizes an arbitrary totally compatible automaton.

\bibliographystyle{plain}
\bibliography{Bibliography}
\end{document}